\documentclass{article}

\makeatletter
\renewcommand\@fnsymbol[1]{\ifcase#1\or 1\or 2\or 3\or 4\or 5\or 6\or 7\or 8\or 9\fi}
\makeatother

\usepackage[preprint]{neurips_2026}

\usepackage[utf8]{inputenc}
\usepackage[T1]{fontenc}
\usepackage{hyperref}
\usepackage{url}
\usepackage{booktabs}

\usepackage{amsfonts}
\usepackage{nicefrac}
\usepackage{microtype}
\usepackage{xcolor}
\usepackage{graphicx}
\usepackage{amsmath}
\usepackage{amsthm}
\usepackage{amssymb}
\usepackage{wrapfig}
\usepackage{subcaption}

\newtheorem{theorem}{Theorem}

\newtheorem{assumption}[theorem]{Assumption} 

\newtheorem*{remark}{Remark}

\usepackage{algorithm}
\usepackage{algpseudocode}

\title{CoCre-Sam: Collective Langevin Sampling Implicitly Fuses Language Models on the Ouija Board}

\author{%
  Tadahiro Taniguchi\thanks{Graduate School of Informatics, Kyoto University, Japan}%
                    \thanks{Ritsumeikan University, Japan} \\
  \And
  Masatoshi Nagano\footnotemark[1] \\
  \And
  Haruumi Omoto\footnotemark[1] \\
  \And
  Yoshiki Hayashi\thanks{University of Reading, UK} \\
}
\date{}

\begin{document}

\maketitle

\begin{abstract}
Collective human activities such as the Ouija board (Kokkuri-san) produce coherent linguistic sequences that no single participant intended. While psychology attributes the planchette's motion to the \emph{ideomotor effect}, the computational mechanism by which decentralized, implicit linguistic knowledge fuses into structured output remains unexplained. We introduce \textbf{CoCre-Sam} (Collective-Creature Sampling), a framework that models each participant as an agent whose internal language model defines an energy landscape and the shared planchette as a particle following collective Langevin dynamics. The central contribution is a computational account of collective sense-making: CoCre-Sam explains how decentralized, embodied micro-actions can implicitly fuse participants' internal language models. This account is grounded by two formal results: (i) the collective dynamics are equivalent to an Unadjusted Langevin Algorithm on the sum of individual energy landscapes (Theorem~\ref{thm:equivalence}); and (ii) the stationary distribution converges to a tempered \emph{Product of Experts} of the agents' language models (Theorem~\ref{thm:sampling_correctness}). Simulations confirm that CoCre-Sam emergently generates plausible character sequences whose semantic distribution lies between those of the individual agents, and that stochasticity is essential---a deterministic counterpart collapses to a single mode while excessive noise degrades coherence. CoCre-Sam thus offers a computational account of how embodied collective interaction \emph{implicitly} fuses language models, yielding testable predictions for future human behavioural studies of collective sense-making.
\end{abstract}

\section{Introduction}
\label{sec:intro}

\begin{figure*}[t]
\centering
\begin{minipage}{0.55\textwidth}
  \centering
  \includegraphics[width=\textwidth]{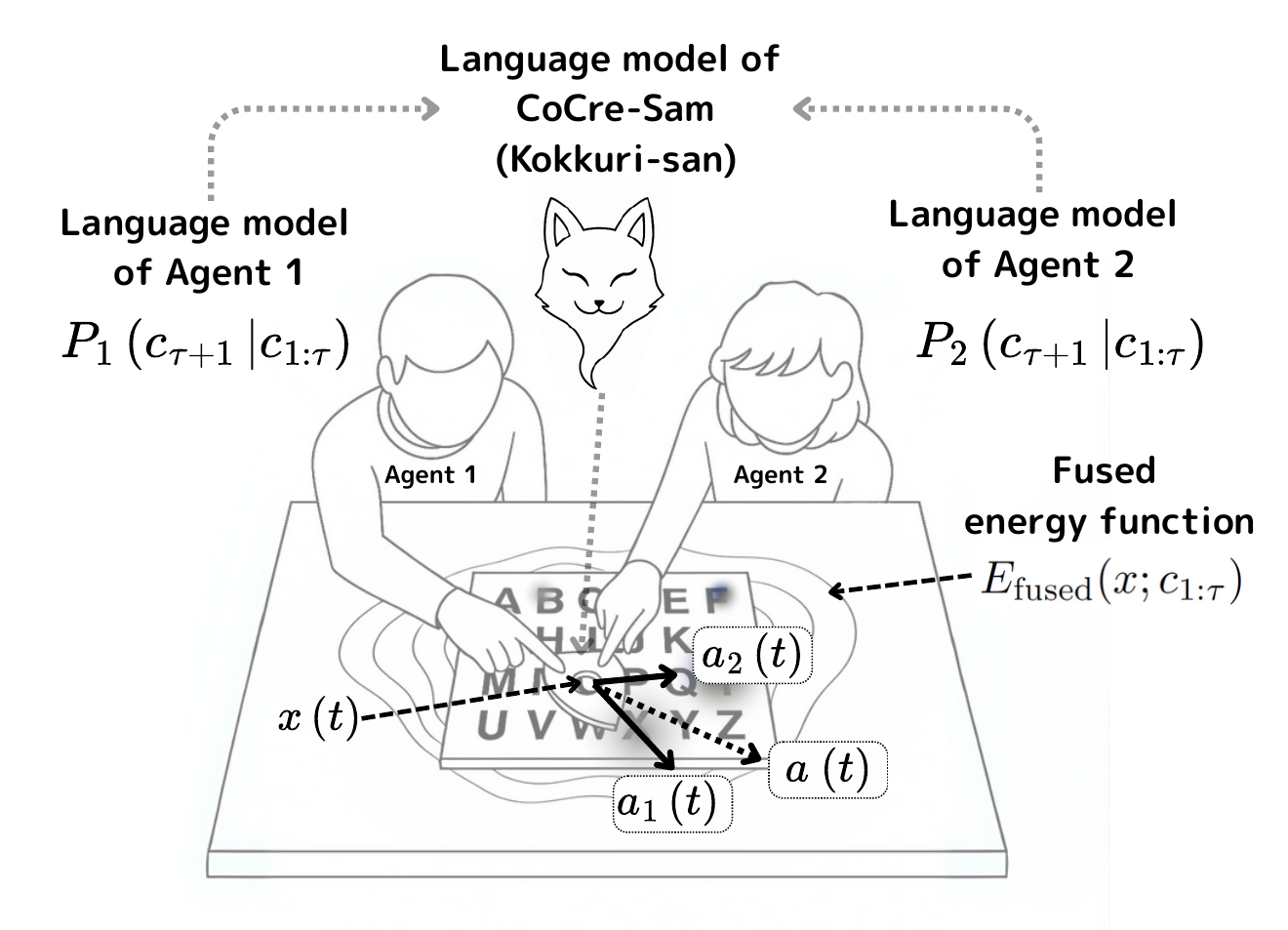}
\end{minipage}%
\begin{minipage}{0.40\textwidth}
  \centering
  \fbox{%
    \begin{minipage}[t]{0.95\textwidth}
      \small
      \textbf{Resulting character sequence:} \\
      \begin{align*} \small
      &\text{Agent 1: } && \bar{P}_{1}(c_{\tau+1} | c_{1:\tau}) \\
      &\text{Agent 2: } && \bar{P}_{2}(c_{\tau+1} | c_{1:\tau}) \\
      &\text{Kokkuri-san: } && \bar{P}_{\mathrm{fused}}(c_{\tau+1} | c_{1:\tau}) \\
      & && \hspace*{-2em}\approx \propto \prod_i \bar{P}_{i}(c_{\tau+1} | c_{1:\tau})^{\frac{\mathcal{T}_i}{\sum_j \mathcal{T}_j }} 
      \end{align*}
    \end{minipage}
  }
\end{minipage}
\caption{Conceptual illustration of the CoCre-Sam framework. Multiple agents interact with a shared planchette on a letter board. Each agent $i$ possesses an internal language model $P_i(c_{\tau+1}|c_{1:\tau})$ defining an energy landscape $E_i$. They exert micro-forces $a_i(t)$ based on their local energy gradient $-\nabla E_i$ and stochastic noise (Langevin dynamics). The collective force $a(t) = \sum_i a_i(t)$ guides the planchette $x(t)$, effectively performing sampling from the fused energy landscape $E_{\rm fused} = \sum_i E_i$. As a result, a character sequence is generated not from either agent’s distribution alone, but from a fused probabilistic model that approximates a product-of-experts form. This leads to the attribution of agency to an imaginary creature named ``Kokkuri-san,'' distinct from either agent, as the perceived author of the outcome.}
\label{fig:concept}
\end{figure*}

Science often seeks rational explanations for phenomena perceived as mysterious or even paranormal.
This paper explores the emergent, and often surprisingly coherent, linguistic patterns that arise during collective human activities, such as the Ouija board~\cite{Carpenter1852,Andersen2019} and its Japanese counterpart, \textit{Kokkuri-san}~\cite{Foster_2006}.
The phenomenon involves multiple participants lightly touching a shared pointer (planchette or coin) which moves across a board of letters to spell out answers, sometimes forming coherent sequences seemingly unintended by any individual (Figure~\ref{fig:concept}).
In this mysterious process, it is often said that the coin is controlled by an imaginary fox-like spirit, a so-called \textit{collective creature} referred to as \textit{Kokkuri-san}, that emerges through human collective behavior.

From a scientific standpoint, this movement is typically explained by the \textit{ideomotor effect}~\cite{Carpenter1852,Gauchou2012,Andersen2019,WegnerWheatley1999}, which attributes the movement of the planchette to unconscious muscular activity driven by participants' expectations; recent predictive-processing accounts further frame Ouija board sessions as a form of active inference under shared priors~\cite{Andersen2019,Friston2010FEP}.
While this psychological perspective offers a scientific basis, it is insufficient to account for the \textit{computational mechanisms} underlying the fusion of participants' implicit linguistic knowledge and the emergence of coherent linguistic structures---a question that lies at the heart of symbol emergence research~\cite{Taniguchi2016SymbolEmergence,Taniguchi2019SymbolEmergence,Taniguchi2024CPC}.
Such accounts fail to explain the \textit{generative process} by which grammatically well-formed sentences occasionally emerge, giving the illusion of intentionality.
Beyond individual psychology, this phenomenon can be framed as a form of \emph{collective intelligence} (CI)~\cite{Surowiecki2004}, and connects to the active study of multi-agent emergent communication in modern machine learning~\cite{Foerster2016LearningCommunicate,LazaridouBaroni2020EmergentSurvey}.
However, unlike classic examples of CI that aggregate pre-existing knowledge, the Ouija board appears to \textit{generate} novel, structured information.
This raises a crucial, unanswered question: what are the underlying \textit{computational dynamics} that allow decentralized, implicit knowledge to be fused into a coherent, emergent output? 

This research gap motivates our central questions: How can multiple agents (human participants), each potentially possessing different internal linguistic models or beliefs (internal knowledge structures now often modeled using Large Language Models, LLMs~\cite{Vaswani2017Attention, Brown2020GPT3, Zhao2023LLMSurvey}), implicitly negotiate and form a collective trajectory through shared physical interaction? Can this complex process be formally understood as a form of \emph{stochastic sampling} from an implicitly \emph{fused} probability distribution representing the collective linguistic knowledge?

To address these questions, we introduce \emph{CoCre-Sam (Collective-Creature Sampling)}, a novel theoretical framework that models the Ouija board/Kokkuri-san phenomenon as \emph{collective Langevin dynamics sampling from fused language models}.

As illustrated conceptually in Figure~\ref{fig:concept}, we model participants as agents, each endowed with an internal (e.g., character-based) language model (internal knowledge structures now often modeled using LLMs).
This model implicitly defines an \textit{energy landscape} over the letter board space given the current context $c_{1:\tau}$ (Eq.~\eqref{eq:effective_energy}).
Each agent generates micro-actions (forces) based on the local gradient at the planchette's shared, visible position $x(t)$, perturbed by stochastic noise \cite{Langevin1908,Lemons1997Langevin}.
The \emph{collective force} (Eq.~\eqref{eq:basic_dynamics}) drives the movement, and this new position $x(t+1)$ serves as continuous visual feedback for calculating the next action (Eq.~\eqref{eq:agent_action}).
We propose that this emergent feedback loop functions as a sampling process (i.e., Langevin Monte Carlo (LMC)~\cite{Neal2011MCMC}) on the \textit{fused} energy landscape (Eq.~\eqref{eq:fused_energy_tau}).
This process continues until a character is selected, at which point the context $c_{1:\tau}$ is updated, thereby \textit{reshaping} the energy landscape for the next sampling step.

We emphasize that our primary goal is not to propose a new engineering method for the fusion of LLMs, but rather to offer a novel computational and mathematical perspective on an existing, intriguing socio-cultural phenomenon.
\paragraph{Contributions.}
The main contributions of this paper are as follows:
\begin{itemize}
    \item \textbf{Computational account of collective sense-making.} We formalize the Ouija board (Kokkuri-san) phenomenon as a collective sampling process in which decentralized, embodied micro-actions implicitly fuse participants' internal language models.
    \item \textbf{Formal grounding of implicit model fusion.} We show that the resulting collective dynamics are equivalent to an Unadjusted Langevin Algorithm on the sum of individual energy landscapes, whose stationary distribution corresponds to a tempered Product of Experts of the agents' language models.
    \item \textbf{Simulation-based validation and predictions.} Through character-level simulations and noise ablations, we demonstrate that CoCre-Sam generates meaningful fused sequences, that the generated semantic distribution lies between those of the individual agents, and that stochasticity is essential for maintaining both diversity and coherence.
\end{itemize}

Full theoretical derivations, detailed methodological parameters, and supplementary data are provided in the appendix.

\section{Related Work}
\label{sec:related_work}

\paragraph{Language model fusion.}
Combining multiple language models has emerged as a major theme in modern NLP. \emph{Parameter-space} methods include weight averaging across fine-tuned checkpoints (Model Soups~\cite{Wortsman2022ModelSoups}), Fisher-weighted merging~\cite{Sharma2024FisherMask}, and task arithmetic~\cite{Ilharco2023TaskArithmetic}. \emph{Output-space} methods such as SpecFuse~\cite{Meng2023SpecFuse} and FuseLLM~\cite{Wan2024FuseLLM} ensemble next-segment predictions or transfer knowledge across LLMs, while Mixture-of-Experts architectures~\cite{Fedus2022Switch,Dou2024MoESurvey} dynamically route inputs to specialized experts. The Metropolis-Hastings Captioning Game~\cite{Matsui2025MHCG} similarly fuses vision-language models through decentralized Bayesian inference. All of these methods propose \emph{explicit} fusion mechanisms, extending a long lineage of model combination from knowledge distillation~\cite{Hinton2015Distill} to the original Product-of-Experts framework~\cite{Hinton2002PoE}; CoCre-Sam, in contrast, shows that fusion can \emph{emerge} as a by-product of physical interaction among agents, with a tempered Product-of-Experts~\cite{Hinton2002PoE} as the stationary distribution (Theorem~\ref{thm:sampling_correctness}).

\paragraph{Langevin dynamics for sampling.}
Langevin dynamics, originally introduced in physics~\cite{Langevin1908,Lemons1997Langevin}, provides a powerful framework for sampling from energy-based distributions~\cite{LeCun2006EBM,Neal2011MCMC}. Discrete approximations such as the Unadjusted Langevin Algorithm (ULA)~\cite{RobertsTweedie1996,Dalalyan2017ULA} and Stochastic Gradient Langevin Dynamics (SGLD)~\cite{WellingTeh2011SGLD} have become standard tools in machine learning for Bayesian inference, and underlie modern score-based and diffusion generative models~\cite{SongErmon2019ScoreBased,Ho2020DDPM} as well as deep energy-based models~\cite{DuMordatch2019EBM}. CoCre-Sam deploys this machinery in a fundamentally different setting: rather than a single agent navigating a known energy landscape, multiple agents with private landscapes interact through a shared planchette, and Theorem~\ref{thm:equivalence} shows that the resulting dynamics is mathematically equivalent to a centralized ULA on the sum of individual landscapes.

\paragraph{Collective intelligence and embodied interaction.}
Classical collective intelligence accounts~\cite{Surowiecki2004} typically \emph{aggregate} pre-existing discrete opinions. Embodied physical interaction is also known to enhance dyadic motor learning~\cite{Ganesh2014,Mireles2017,Saracbasi2021Haptic,Thorne2019HapticStructure}. CoCre-Sam unifies these threads from a different angle: it shows how decentralized, implicit linguistic knowledge can be \emph{aggregated and generated} through embodied physical interaction, producing structured symbolic outputs that no single participant intended. Further background is provided in Appendix~\ref{app:supplementary_background}.

\paragraph{Multi-agent emergent communication.}
A separate but related line of work studies how communication protocols and shared symbol systems can emerge among learning agents~\cite{Foerster2016LearningCommunicate,LazaridouBaroni2020EmergentSurvey}. CoCre-Sam is naturally read in this tradition---it captures emergent symbolic output through embodied interaction rather than discrete message passing. Recent work on co-creative learning between humans and AI~\cite{Okumura2025CoCreative} explores closely related Metropolis--Hastings interaction dynamics, while decentralized collective world models~\cite{Nomura2025DCWM} show how multi-agent world models can ground shared representations through interaction; both provide complementary algorithmic perspectives on the questions CoCre-Sam addresses through Langevin sampling.

\section{Framework: CoCre-Sam}
\label{sec:model}

\subsection{Ouija Board Dynamics}
\label{subsec:Ouija_Board_Dynamics}

We model the Ouija board/Kokkuri-san process as a collective dynamical system.
We consider a set of $N$ agents (participants), $i \in \mathcal{I} = \{1, \dots, N\}$.
The primary state of the system at discrete time $t$ is the position of the shared planchette, $x(t) \in \mathbb{R}^d$, on a 2-dimensional board ($d=2$).
The board displays a set of possible discrete outputs $\mathcal{C}$ (e.g., alphabet letters), where each symbol $c \in \mathcal{C}$ is associated with a specific target location (goal) $g_c \in \mathbb{R}^d$.
The fundamental dynamic involves each agent $i$ contributing an action $a_i(t) \in \mathbb{R}^d$ (a micro-force or intended movement). The planchette's position updates based on the collective effect of these actions:
\begin{equation}
 x(t+1) = x(t) + \sum_{i \in \mathcal{I}} a_i(t).
 \label{eq:basic_dynamics}
\end{equation}
How each $a_i(t)$ is generated is detailed later. The high-level process of generating character sequences from these dynamics is detailed in Appendix~\ref{app:supplementary_methods}.

\subsection{Energy Landscape from Language Models on the Ouija Board}
\label{subsec:energy_landscape}

A key conceptual shift in our approach is viewing the Ouija board not just as a physical surface, but as an embedding space where characters are located.
The core assumption is that each agent $i$ possesses an internal language model, $P_i(c_{\tau+1} | c_{1:\tau})$, which assigns a probability to the next possible letter $c_{\tau+1} \in \mathcal{C}$ given the current context $c_{1:\tau}$.
This probabilistic model is realized through an individual energy function, akin to energy-based models (EBMs)~\cite{LeCun2006EBM}. For a given target goal $g_c$, the elemental energy is:
\begin{equation}
    E_i(x; g_c) = -\phi(\|x - g_c\|) - \phi_{0}
    \label{eq:elemental_energy}
\end{equation}
where $\phi(r)$ is a potential function that increases with the distance $r$ from the goal.

While the full model involves a hierarchical process (detailed in Appendix~\ref{app:supplementary_methods}), in this paper we focus on the case ($\triangle T = 1$) where agents' movements are governed by an \textit{effective energy function} that averages over all possible goals, weighted by their probabilities:
\begin{equation}
    E_{i}(x \mid c_{1:\tau}) = \sum_{c \in \mathcal{C}} P_i(c_{\tau+1}=c \mid c_{1:\tau}) E_i(x; g_c).
    \label{eq:effective_energy}
\end{equation}
The collective behavior emerges from a \textit{fused energy function}, defined as the sum of these individual effective landscapes:
\begin{equation}
E_{\text{fused}}(x \mid c_{1:\tau}) = \sum_{i \in \mathcal{I}} E_{i}(x \mid c_{1:\tau}).
\label{eq:fused_energy_tau}
\end{equation}
This fused landscape $E_{\text{fused}}$ maps the combined next-character predictions of all agents onto the board space.

\subsection{Agent Actions Determined by Langevin Dynamics}
\label{subsec:agent_actions}

We now specify how each agent's action $a_i(t)$ in Eq.~\eqref{eq:basic_dynamics} is generated, linking it to their internal energy landscape $E_i$.
We model each agent $i$'s action $a_i(t)$ as a step consistent with Langevin dynamics operating on their own energy function $E_i(x| c_{1:\tau})$.
This action represents the micro-force the agent applies, combining a tendency to move towards lower personal energy regions with inherent stochasticity:
\begin{equation}
    a_i(t) = - \eta \nabla_x E_i(x(t) \mid c_{1:\tau}) + \sqrt{2 D_i} \xi_i(t)
    \label{eq:agent_action}
\end{equation}
Here, $\eta > 0$ is a step size parameter scaling the force derived from the negative gradient $-\nabla_x E_i$ of agent $i$'s energy landscape. 
The term $D_i \ge 0$ represents the intensity of stochastic noise (or diffusion coefficient) specific to agent $i$, and $\xi_i(t) \sim \mathcal{N}(0, I_d)$ is standard Gaussian noise.
We can also associate an individual temperature $\mathcal{T}_i = D_i / \eta$ with each agent.
While each agent acts based on their own energy landscape and noise, their actions are collectively summed (via Eq.~\eqref{eq:basic_dynamics}) to produce the emergent movement of the planchette.

\section{Theoretical Analysis}
\label{sec:theory}
This section provides the theoretical grounding for the CoCre-Sam framework, demonstrating that the collective dynamics mathematically correspond to a principled sampling process. Under standard assumptions detailed in Appendix~\ref{app:supplementary_theoretical_analysis} (Assumptions~\ref{app:assump:smoothness}--\ref{app:assump:step_size}), we show that the dynamics emerging from decentralized agent actions are equivalent to a standard Langevin MCMC algorithm operating on the fused energy landscape. This equivalence ensures correct sampling properties.

Theorem~\ref{thm:equivalence} hinges on the fact that the planchette provides a single shared physical state $x(t)$ through which all agents are coupled. Because each agent evaluates its local force at the same position $x(t)$, the sum of decentralized gradient terms becomes the gradient of the summed energy landscape. Thus, the shared planchette is not merely an interface for interaction, but the physical coupling mechanism that allows decentralized micro-actions to be interpreted as a single Langevin update on the fused landscape. If agents acted on separate private states, this reduction to a centralized ULA would no longer follow.

Our first key theoretical result establishes that the complex, decentralized micro-actions of the agents (Eq.~\eqref{eq:agent_action}) mathematically converge into a single, well-defined centralized process. Theorem~\ref{thm:equivalence} formally shows that the decentralized CoCre-Sam process can be interpreted as implementing a centralized Unadjusted Langevin Algorithm (ULA) operating on the fused energy landscape.

\begin{theorem}[\bf CoCre-Sam Dynamics as ULA on Fused Energy]
\label{thm:equivalence}
The discrete-time dynamics of the planchette position $x(t)$, governed by the combination of the basic collective update rule (Eq.~\eqref{eq:basic_dynamics}) and the individual Langevin-based actions (Eq.~\eqref{eq:agent_action}), constitutes an instance of the ULA~\cite{RobertsTweedie1996, WellingTeh2011SGLD}. Specifically, it corresponds to ULA (defined in Eq.~\eqref{app:eq:ula_general} in Appendix~\ref{app:supplementary_background})
targeting the fused energy function $E(x) = E_{\rm fused}(x \mid c_{1:\tau})$, with step size $\eta$ and effective temperature $\mathcal{T}_{\rm fused} = D_{\rm fused} / \eta$, where $D_{\rm fused} = \sum_{i \in \mathcal{I}} D_i$.
\end{theorem}

Building directly upon the equivalence established above, we state the main result regarding the sampling properties. Theorem~\ref{thm:sampling_correctness} provides the core theoretical justification for CoCre-Sam. It guarantees that the emergent dynamics converge to a well-defined probabilistic state reflecting the agents' combined knowledge, allowing the planchette's long-term behavior to be interpreted as sampling from the fused model. The effective temperature $ \mathcal{T}_{fused} = D_{\rm fused}/\eta$ controls the exploration level.

\begin{theorem}[\bf Stationary Distribution and Sampling Correctness]
\label{thm:sampling_correctness}
Under the assumptions detailed in Appendix~\ref{app:supplementary_theoretical_analysis} (Assumptions~\ref{app:assump:smoothness}--\ref{app:assump:step_size}), the stochastic process $\{x(t)\}_{t \ge 0}$ generated by the CoCre-Sam dynamics (defined by Eq.~\eqref{eq:basic_dynamics} and Eq.~\eqref{eq:agent_action}) for a fixed context $c_{1:\tau}$ is ergodic, and its unique stationary distribution is the Gibbs-Boltzmann distribution $P^{\rm stat}_{\rm fused}(x) \propto \exp(-E_{\rm fused}(x \mid c_{1:\tau})/\mathcal{T}_{\rm fused})$, where $\mathcal{T}_{\rm fused}=D_{\rm fused}/\eta$. Consequently, CoCre-Sam samples from a target distribution reflecting the fused models.
\end{theorem}

The convergence to this stationary distribution relies on the established theoretical guarantees for ULA under appropriate conditions~\cite{RobertsTweedie1996, Dalalyan2017ULA} (see Appendix~\ref{app:supplementary_theoretical_analysis} for proof details).

\begin{remark}[Connection to the Product of Experts]
\label{rmk:poe}
Theorem~\ref{thm:sampling_correctness} exposes an explicit link between
CoCre\text{-}Sam and the \emph{Product of Experts} (PoE) framework.
The stationary distribution $P^{\rm stat}_{\rm fused}(x)$ can be shown to be proportional to a tempered product of the stationary distributions of individual agents (see Appendix~\ref{app:subsec:poe_derivation} for derivation), approximating $P^{\rm stat}_{\rm fused}(x) \propto \prod_{i \in \mathcal{I}} \{P^{\rm stat}_{\rm i}(x)\}^{1/N}$ (where $N=|\mathcal{I}|$ is the number of agents) in the special case of equal temperatures.
\end{remark}

This PoE structure extends from the continuous dynamics of the planchette to the discrete selection of characters. The stationary distribution of the planchette's position, $P^{\rm stat}_{\rm fused}(x)$, concentrates the system in the low-energy regions of the fused landscape, which represent a consensus among the agents' models. The character selection process (Algorithm~\ref{app:alg:basic_ouija_dynamics} in Appendix~\ref{app:supplementary_methods}) is a quantization of this continuous state, where the character whose goal is nearest to the planchette's trajectory accumulates the most votes. Consequently, the probability of selecting a character $c_{\tau+1}$ through this voting process, denoted $\bar{P}_{\rm fused}(c_{\tau+1}\mid c_{1:\tau})$, approximates a tempered Product of Experts, $\bar{P}_{\rm fused}(c_{\tau+1}|c_{1:\tau}) \propto \prod_{i \in \mathcal{I}} \{\bar{P}_{\rm i}(c_{\tau+1}|c_{1:\tau})\}^{\mathcal{T}_i/\mathcal{T}_{\rm fused}}$, where $\bar{P}_{i}$ denotes the probability that the $i$-th agent generates $c_{\tau+1}$ through Ouija board dynamics when acting alone.

\section{Experiments}
\label{sec:results}

\subsection{Experimental Setup}
\label{sec:experimental_setup}
We instantiate the CoCre-Sam framework with $N{=}2$ agents whose internal language models are character-level $n$-gram models ($n{=}6$) trained on $10^5$ flower-name samples each, drawn with opposite colorfulness biases from a 104-flower vocabulary (a curated A--Z pool). The Langevin parameters are $\eta{=}0.1$ and $D_i{=}0.01$, giving an effective fused temperature $\mathcal{T}_{\rm fused}{=}0.2$; the inner Langevin loop runs $T_{\rm inner}{=}2000$ steps, with the last $5\%$ of each trajectory voting for the next character. Sequences terminate at \texttt{EOS} or after 20 characters. Unless stated otherwise, all statistics are reported as mean $\pm$ SD across $5$ random seeds with $100$ trials per (seed, condition). Full details of the environment, language-model training, energy potential (Cauchy loss with $r_0{=}0.3$), and voting are deferred to Appendix~\ref{app:supplementary_methods}.

\subsection{Convergence: Sampling Dynamics on the Fused Energy Landscape}
\label{sec:results_convergence}

\begin{figure*}[t]
  \centering
    \begin{minipage}{\textwidth}
        \centering
        \begin{minipage}{0.48\textwidth}
          \centering
          \includegraphics[width=\linewidth]{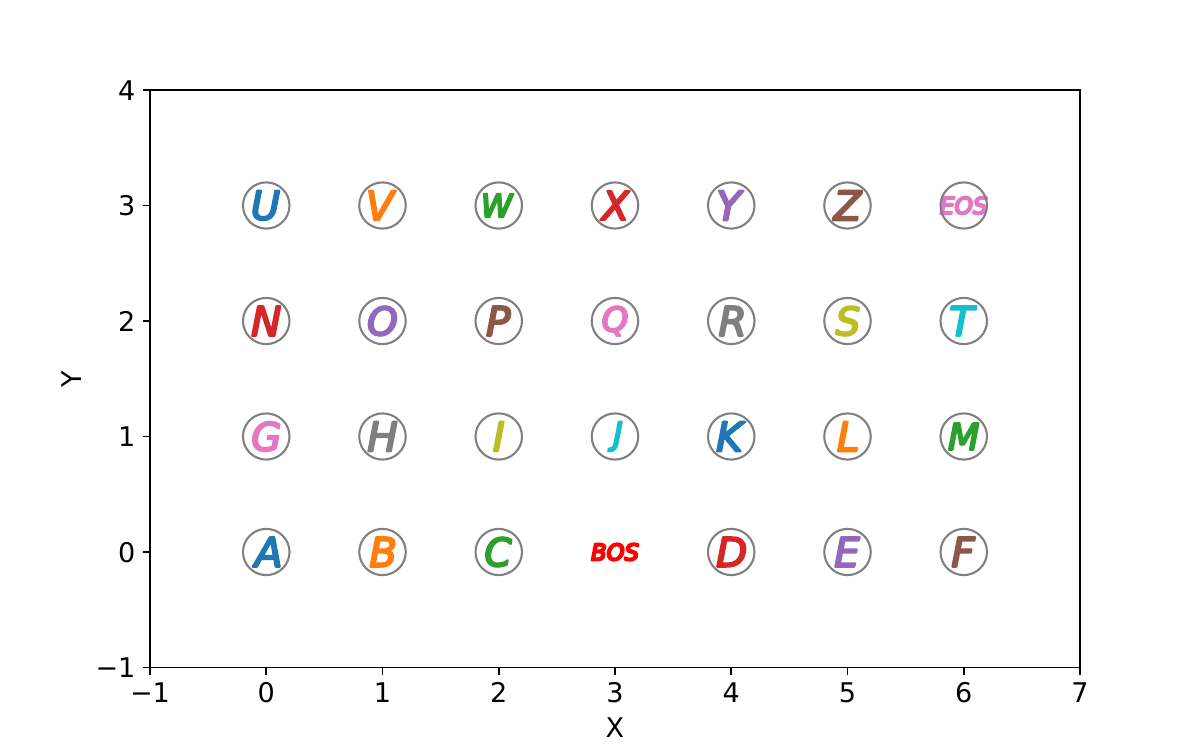}
        \end{minipage}%
        \hfill
        \begin{minipage}{0.51\textwidth}
          \centering
          \includegraphics[width=\linewidth]{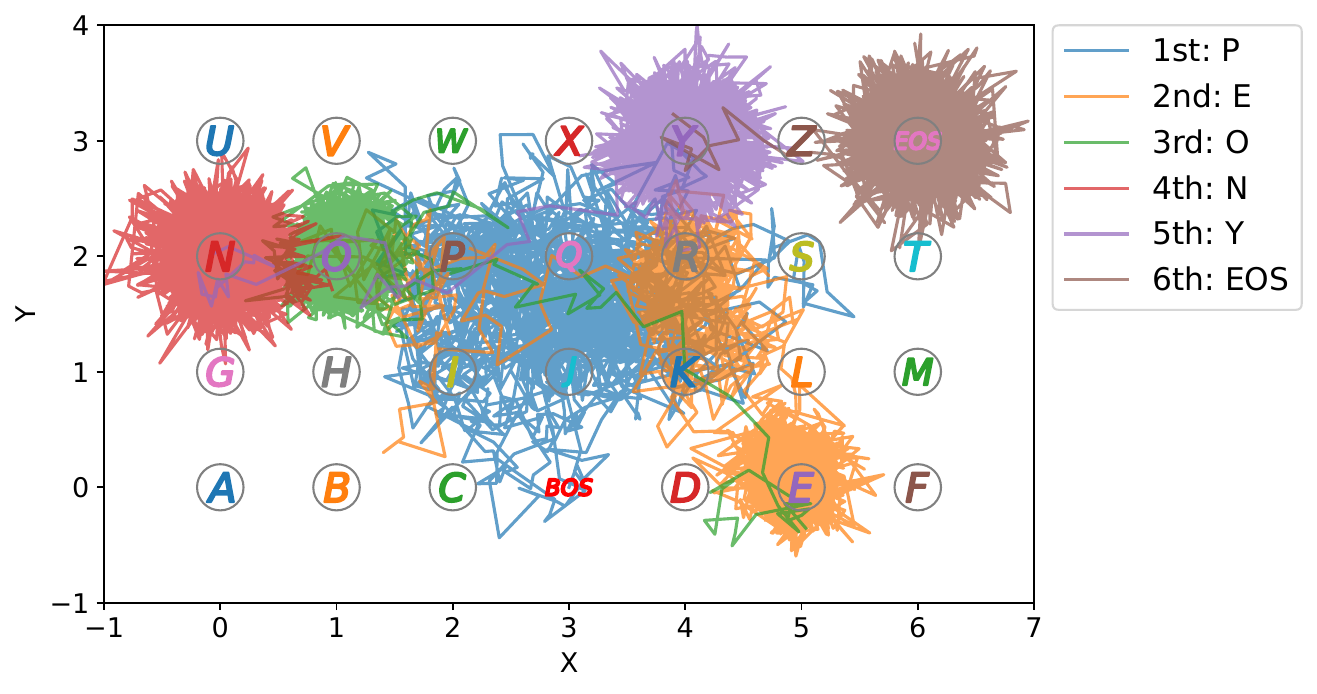}
        \end{minipage}
\caption{\small \textbf{(Left)} The CoCre-Sam simulation environment, representing the Ouija board as a 2D embedding space. Each letter, plus 'BOS' (Begin) and 'EOS' (End) tokens, is assigned a specific target location $g_c$. \textbf{(Right)} An example trajectory of the planchette $x(t)$ from a simulation with $N=2$ agents. The collective Langevin dynamics, driven by the agents' fused energy landscapes (Eq. (4)), stochastically sample the coherent sequence 'P-E-O-N-Y-EOS'. Each color represents the trajectory segment as the dynamics converge on the goal location for each successive character, demonstrating the model's generative process.}
  \label{fig:ouija_combined}
    \end{minipage}
\vspace{-5mm}
\end{figure*}

We verify the CoCre-Sam dynamics by simulating the planchette's trajectory $x(t)$ on the fused energy landscape $E_{\text{fused}}(x \mid c_{1:\tau})$ (Fig. 2, right, shows an example trajectory for generating the word \texttt{PEONY}). We first verify that the system behaves according to the proposed theoretical model. A detailed, step-by-step visualization of the sampling process for \texttt{PEONY} is provided in Fig.~\ref{app:fig:generation_sequence_peony} (Appendix~\ref{app:supplementary_figures}), illustrating the planchette's stochastic navigation of the evolving energy landscape.

\subsection{Knowledge Fusion: Emergent Product of Experts}
\label{sec:results_fusion}

\begin{table}[bt]
 \centering
 \caption{Top 5 generated flower names and their frequencies under various conditions. CoCre-Sam ($\propto \sqrt{P_1P_2}$) is the proposed model. 'No Noise' ($\mathcal{T}_i=0.0$) is a deterministic gradient descent, 'Middle Noise' ($\mathcal{T}_i=0.5$) an intermediate setting, and 'High Noise' ($\mathcal{T}_i=1.0$) approaches a random walk. Frequencies are reported as \emph{mean $\pm$ SD across 5 random seeds}, with 100 trials per (seed, condition).}
 \label{tab:top5_flower_results_en_expanded}
   \resizebox{\linewidth}{!}{%
 \begin{tabular}{l|cc|cc|cc}
  \toprule
   & \multicolumn{2}{c|}{\textbf{Agent 1 ($P_1$)}} & \multicolumn{2}{c|}{\textbf{Agent 2 ($P_2$)}} & \multicolumn{2}{c}{\textbf{CoCre-Sam ($\propto \sqrt{P_1P_2}$)}} \\
  \# & String & Freq (mean $\pm$ SD) & String & Freq (mean $\pm$ SD) & String & Freq (mean $\pm$ SD) \\
  \midrule
  1 & petunia   & $21.6 \pm 3.9$ & quassia   & $27.6 \pm 5.2$ & quaker     & $23.6 \pm 1.6$ \\
  2 & jacobinia & $16.0 \pm 2.6$ & quaker    & $17.6 \pm 3.4$ & jacobinia  & $17.6 \pm 2.9$ \\
  3 & quince    & $12.0 \pm 2.4$ & primrose  & $\phantom{0}9.8 \pm 4.3$ & petunia    & $15.8 \pm 3.5$ \\
  4 & iris      & $11.2 \pm 2.3$ & rue       & $\phantom{0}9.6 \pm 2.4$ & rockrose   & $12.4 \pm 3.0$ \\
  5 & kalmia    & $\phantom{0}8.8 \pm 1.6$ & jacobinia & $\phantom{0}8.2 \pm 1.6$ & quassia    & $\phantom{0}8.4 \pm 4.5$ \\
  \bottomrule
 \end{tabular}}

 \bigskip

   \resizebox{\linewidth}{!}{%
 \begin{tabular}{l|cc|cc|cc}
  \toprule
   & \multicolumn{2}{c|}{\textbf{No Noise ($\mathcal{T}_i=0.0$)}} & \multicolumn{2}{c|}{\textbf{Middle Noise ($\mathcal{T}_i=0.5$)}} & \multicolumn{2}{c}{\textbf{High Noise ($\mathcal{T}_i=1.0$)}} \\
   \# & String & Freq (mean $\pm$ SD) & String & Freq (mean $\pm$ SD) & String & Freq (mean $\pm$ SD) \\
   \midrule
  1 & quince     & $80.0 \pm 40.0$ & rockrose  & $\phantom{0}2.8 \pm 0.4$ & c         & $\phantom{0}1.0 \pm 1.1$ \\
  2 & quaker     & $20.0 \pm 40.0$ & jacobinia & $\phantom{0}2.4 \pm 1.4$ & u         & $\phantom{0}0.8 \pm 0.7$ \\
  3 & --         & --              & daisy     & $\phantom{0}2.4 \pm 1.5$ & a         & $\phantom{0}0.6 \pm 0.5$ \\
  4 & --         & --              & quassia   & $\phantom{0}1.8 \pm 0.7$ & x         & $\phantom{0}0.6 \pm 0.8$ \\
  5 & --         & --              & rue       & $\phantom{0}1.6 \pm 1.4$ & f         & $\phantom{0}0.6 \pm 0.5$ \\
   \bottomrule
 \end{tabular}}

 \smallskip
 \footnotesize
 \emph{Note:} For deterministic ``No Noise,'' the planchette collapses to a single attractor that varies between \texttt{quince} and \texttt{quaker} across seeds, hence the bimodal $80\pm40$ / $20\pm40$ split. ``Mid'' and ``High'' Noise frequently emit garbage strings or single characters; the rows show the most frequent valid words under those settings.
\end{table}

\begin{table}[bt]
\centering
\small
\caption{Perplexity of generated valid words (mean $\pm$ SD across 5 random seeds, 100 trials each; word-level, then averaged with denominator 100 to match the generative process). Each row indicates the generating model; each column indicates the evaluation LM. Lower perplexity means higher predictability. The lowest entry in each row is \textbf{bold}; for CoCre-Sam, $P_2$ and Fused are both bolded because they are statistically tied (paired $t$-test, $|t|=0.64$, $p>0.40$, $n=5$ seeds).}
\label{tab:perplexity}
\begin{tabular}{l|ccc}
\toprule
& \multicolumn{3}{c}{Evaluated by LM:} \\
\textbf{Gen. by} & $P_1$ & $P_2$ & Fused ($\sqrt{P_1P_2}$) \\
\midrule
Agent 1   & $\mathbf{1.840 \pm 0.027}$ & $2.128 \pm 0.076$ & $1.920 \pm 0.038$ \\
Agent 2   & $2.509 \pm 0.133$ & $\mathbf{1.935 \pm 0.046}$ & $2.139 \pm 0.074$ \\
CoCre-Sam & $1.951 \pm 0.037$ & $\mathbf{1.855 \pm 0.020}$ & $\mathbf{1.862 \pm 0.017}$ \\
\bottomrule
\end{tabular}
\vspace{-5mm}
\end{table}
Having verified the basic sampling dynamics (Fig. 2, right; Fig.~\ref{app:fig:generation_sequence_peony}), we investigated whether CoCre-Sam can emergently generate meaningful sequences. Table~\ref{tab:top5_flower_results_en_expanded} presents the top 5 generated flower names under various conditions, showing CoCre-Sam properly samples plausible words based on the underlying models' likelihoods. For instance, while Agent 1 (biased toward colorful flowers) favored words like \textit{petunia}, and Agent 2 (biased toward less colorful ones) favored \textit{quassia}, the CoCre-Sam model effectively fused these preferences, generating outputs reflecting both, such as \textit{quaker} and \textit{jacobinia}. This demonstrates a successful fusion of linguistic knowledge. In contrast, deterministic gradient descent ($D_i = 0$, ``No Noise'') collapses to a single attractor per seed---predominantly \texttt{quince}, with one of the five seeds converging to \texttt{quaker} instead (hence the bimodal $80\pm40$\,/\,$20\pm40$ split in Table~\ref{tab:top5_flower_results_en_expanded})---underscoring the necessity of stochasticity for exploration, while excessive noise (``High Noise'') degrades the output to single-character or garbage strings.

To visualize this fusion on a semantic level, Figure~\ref{fig:kde_distributions} plots the ``colorfulness'' distribution of the generated words. The CoCre-Sam distribution clearly lies between the opposing biases of Agent 1 (more colorful) and Agent 2 (less colorful), providing visual evidence that the collective process generates words from a genuinely fused semantic space.

We evaluated perplexity against each agent's model ($P_1, P_2$) and the theoretical fused model, approximated by the normalized geometric mean $\sqrt{P_1 P_2}/Z$ (with partition function $Z=\sum_c \sqrt{P_1(c) P_2(c)}$ over the character vocabulary) for our $N=2$, equal-temperature setting (denoted ``Fused''). Perplexity measures how well a probability model predicts a sample; a lower score indicates better predictability. As shown in Table~\ref{tab:perplexity}, each agent's own generations are best predicted by its own model (Agent 1 on $P_1$: $1.840 \pm 0.027$; Agent 2 on $P_2$: $1.935 \pm 0.046$), confirming the diagonal pattern. Crucially, \emph{CoCre-Sam's output achieves uniformly low perplexity across all three evaluators} ($P_1$: $1.951 \pm 0.037$; $P_2$: $1.855 \pm 0.020$; Fused: $1.862 \pm 0.017$), with $P_2$ and the fused model statistically tied (paired $t$-test, $p>0.40$, $n=5$ seeds). This balanced perplexity---the generation is equally well-explained by either agent's model and by the implicit Product-of-Experts predicted by Theorem~\ref{thm:sampling_correctness}---is the empirical signature of a successful fusion.

\subsection{Necessity of Stochasticity}
\label{sec:results_stochasticity}

The qualitative pattern in Table~\ref{tab:top5_flower_results_en_expanded} (deterministic ``No Noise'' collapse vs.\ ``High Noise'' degradation) suggests that stochasticity is essential. We confirm this with a finer-grained ablation: Figure~\ref{fig:ablation_noise} sweeps the noise intensity $\mathcal{T}_{\rm fused}$ and reports both the number of valid words generated (reliability) and the entropy of the resulting word distribution (diversity).
The results reveal a clear trade-off between reliability and diversity.
This demonstrates that an optimal level of noise, corresponding to inherent fluctuations in the force applied by participants' fingers (e.g., hand tremor), is crucial for balancing diversity and validity---a quantitative reflection of why deterministic gradient descent cannot replicate the emergent generation we observed above.

\begin{figure*}[t]
  \centering
  \begin{minipage}[t]{0.48\textwidth}
    \centering
    \includegraphics[width=0.9\linewidth]{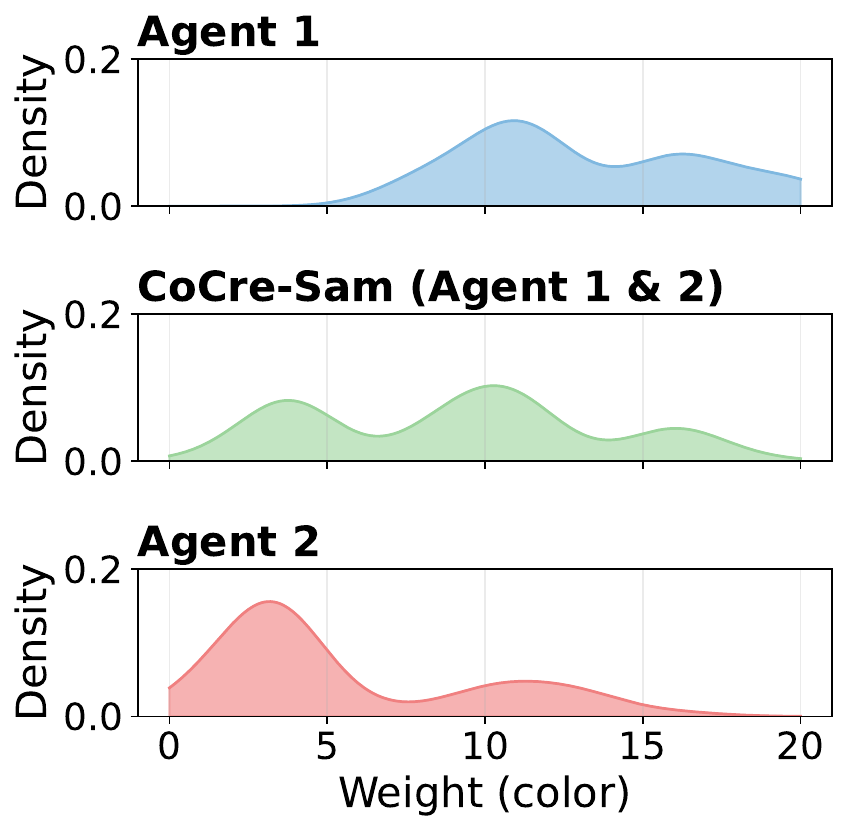}
    \caption{\small Semantic density distributions of generated words. Agent 1 is biased toward high-colorfulness words, Agent 2 toward low-colorfulness words, while CoCre-Sam yields an intermediate fused distribution. Pooled across 5 seeds with 100 trials each.}
    \label{fig:kde_distributions}
  \end{minipage}\hfill
  \begin{minipage}[t]{0.48\textwidth}
    \centering
    \includegraphics[width=0.9\linewidth]{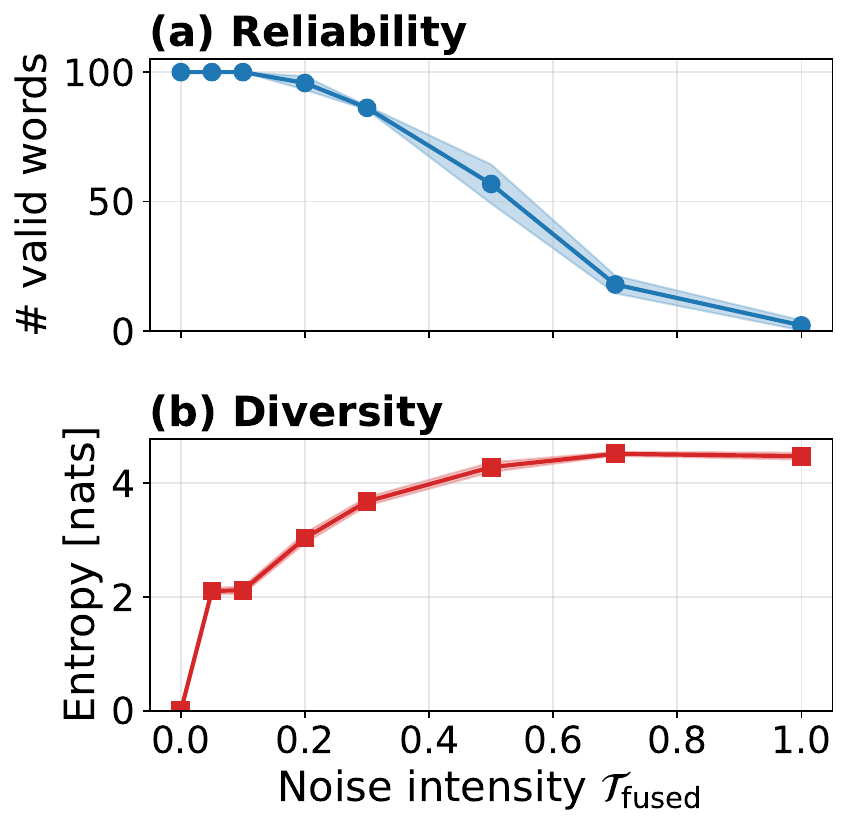}
    \caption{\small Ablation on noise intensity ($\mathcal{T}_{\rm fused}$). Mean $\pm$ SD across 5 seeds with 100 trials each. Reliability decreases as noise increases, while diversity increases toward the finite-sample ceiling $\log 100 \approx 4.61$ nats (i.e., approximately uniform sampling over 100 unique outputs at high noise), illustrating the validity--diversity trade-off.}
    \label{fig:ablation_noise}
  \end{minipage}
\end{figure*}

\section{Discussion}
\label{sec:discussion}

While the ideomotor effect~\cite{Carpenter1852,Gauchou2012,Andersen2019} explains the unconscious motor movements in Ouija board sessions, it doesn't account for the emergence of coherent linguistic sequences. CoCre-Sam offers a computational explanation: structured outputs arise from decentralized sampling over a fused distribution of internal language models.

This framework offers insights into \textit{agency}. Standard models link the \textbf{sense of agency (SoA)} to predicting action outcomes~\cite{yano2020statistical}, often focusing on low-level sensorimotor control, sometimes studied via hierarchical perspectives~\cite{oi2024hierarchical}. The Ouija board presents a distinct case: agency is attributed to a fictional \textit{collective creature} (``Kokkuri-san moved us''), suggesting the emergence of a perceived third-party agent beyond individual motor control.

We speculate this experience parallels the anthropomorphism directed at LLMs. In both scenarios, coherent linguistic output from probabilistic models invites interpretations of intentionality. CoCre-Sam suggests this perceived agency isn't tied to individual control but arises from the coherence of the shared output generated from a fused prior, inviting anthropomorphism.

This invites a computational account of perceived agency that complements standard sensorimotor SoA models. Classical SoA frameworks~\cite{yano2020statistical,oi2024hierarchical} locate agency in the match between predicted and observed sensorimotor consequences of one's own actions: an agent claims authorship of an outcome when its internal forward model accurately anticipates it. Under this paradigm, no individual Ouija-board participant should attribute the produced sequence to themselves---each is aware their hand was perturbed by collective dynamics they did not consciously intend. CoCre-Sam clarifies why this absence of individual-level agency does not extinguish perceived agency but rather \emph{displaces} it: when the output stream itself exhibits the coherence of a fused Product-of-Experts distribution (Theorem~\ref{thm:sampling_correctness}), participants attribute the missing intentionality to an emergent third party (``Kokkuri-san''). Perceived agency thus can arise at the \emph{collective output level} without any agent operating it at the individual sensorimotor level---a regime not captured by either classical ideomotor accounts or standard sensorimotor SoA models.

This emergence of coherent collective output from decentralized implicit knowledge positions the CoCre-Sam framework as a novel mechanism for \textbf{collective intelligence} (CI). It demonstrates how coherent outcomes can emerge from decentralized, noisy knowledge sources~\cite{Surowiecki2004}, differing from typical CI models that rely on explicit communication. CoCre-Sam highlights the role of \textbf{embodied physical interaction} as a medium for integrating implicit knowledge~\cite{Ganesh2014, Mireles2017, Saracbasi2021Haptic, Thorne2019HapticStructure}.

This resonates with theoretical frameworks like Collective Predictive Coding (CPC)~\cite{Taniguchi2024CPC}, which views symbol systems (e.g., language) as shared priors emerging from decentralized interactions. In CPC, symbolic interaction allows agents to align their internal states within a common interpretive space. CoCre-Sam offers a potential physical instantiation of such ideas, where the collective sampling from the fused energy landscape effectively realizes a shared symbolic prior through embodied dynamics.

CoCre-Sam is complementary to, rather than in competition with, active inference accounts of the Ouija board~\cite{Andersen2019,Friston2010FEP}. Active inference explains \emph{why individual agents move}---to minimise prediction error under their own generative models; CoCre-Sam explains \emph{what the collective dynamics computes}---sampling from a fused energy landscape. Together they form a two-level account: active inference at the sensorimotor level governs each participant's micro-actions, while collective Langevin sampling at the emergent output level governs the structured linguistic sequences that arise. CoCre-Sam thus provides the missing generative account that sensorimotor active inference alone cannot supply.

Conceptually, CoCre-Sam also relates to \textbf{shared control} and \textbf{shared autonomy} in human-robot interaction~\cite{dragan2013policy,kirby2007shared,kim2012blended,muraoka2024}: multiple noisy or partial inputs are blended to produce a unified output. Here, however, blending occurs in implicit gradient space over language-model-derived energy landscapes. This suggests a design principle for human--AI co-creation~\cite{Okumura2025CoCreative}: coherent collective output can emerge when each participant reacts to a \emph{shared observable state} while preserving individual stochasticity, without requiring explicit communication or shared parameters.

Implementation-specific effects of the 2D lattice, Cauchy attractor, and finite-time voting are discussed in Appendix~\ref{app:subsec:experimental_conditions}; they point to learned board embeddings as a promising future direction.

\section{Limitations and Future Work}
\label{sec:limitations}

The most important limitation of the present study is that its validation is purely computational: we have not yet collected behavioural data from human Ouija board / Kokkuri-san sessions. Thus, CoCre-Sam should not be read as an empirical confirmation of how human participants actually generate coherent sequences in such settings. Rather, its contribution at this stage is to turn this phenomenon into a testable computational hypothesis. By specifying the latent variables, collective dynamics, and observable predictions of the process, the framework opens a path toward behavioural experiments that can examine whether human collective interaction exhibits the predicted signatures, such as stochasticity-dependent coherence, scaling effects with group size, and displaced sense of agency at the collective-output level.

Our computational validation is also restricted to $N{=}2$ agents with character-level $n$-gram language models. While Theorems~\ref{thm:equivalence} and~\ref{thm:sampling_correctness} hold for arbitrary $N$ and any $P_i$, the empirical scaling behaviour remains to be characterised. In particular, the effective temperature $\mathcal{T}_{\rm fused}{=}\sum_i D_i/\eta$ increases linearly with $N$, which would push the stationary distribution toward uniformity and potentially degrade coherent output for large groups; whether this theoretical prediction matches empirical Ouija-board observations (e.g., larger groups producing less coherent outputs) is an intriguing testable question. Future work should also examine whether new phenomena such as semantic drift or sycophantic collapse appear under large LMs, and whether coherent fusion persists under heterogeneous temperatures or online-adapted internal models. Integrating online updates of $P_i$ (e.g.\ via predictive coding) is a natural extension that connects to the CPC framework discussed above. Implementation aspects of the 2D lattice embedding and finite-time voting that shape the empirical character marginal are discussed in Section~\ref{sec:discussion}.

\section{Conclusion}
\label{sec:conclusion}

We introduced CoCre-Sam, a computational framework that recasts the Ouija-board / Kokkuri-san phenomenon as collective Langevin sampling from a fused energy landscape. Two theoretical guarantees---the equivalence to a centralised Unadjusted Langevin Algorithm on the sum of agent energies (Theorem~\ref{thm:equivalence}) and convergence to a tempered Product of Experts (Theorem~\ref{thm:sampling_correctness})---are matched by a multi-seed empirical validation that exhibits balanced fusion across all three perplexity evaluators (Table~\ref{tab:perplexity}). The framework links individual cognition, embodied collective action, and emergent linguistic phenomena, offering a single computational mechanism for the study of collective intelligence and shared autonomy in human and human--AI systems.

\begin{ack}
This work was supported by JSPS KAKENHI Grant Numbers JP23H04835 and JP21H04904.
The authors declare no competing interests.

\medskip
\noindent\textbf{Code and data availability.}
The source code supporting this publication is available at Zenodo: \url{https://doi.org/10.5281/zenodo.15876240}.
The corresponding GitHub repository is available at: \url{https://github.com/Tanichu-Laboratory/CoCre-Sam}.
\end{ack}

\bibliographystyle{unsrt}
\bibliography{cocre}

\appendix
\renewcommand{\theequation}{\thesection.\arabic{equation}}
\renewcommand{\thefigure}{\thesection.\arabic{figure}}
\renewcommand{\thetable}{\thesection.\arabic{table}}
\renewcommand{\thealgorithm}{\thesection.\arabic{algorithm}}
\setcounter{equation}{0}
\setcounter{figure}{0}
\setcounter{table}{0}
\setcounter{algorithm}{0}

\section{Supplementary Methods}
\label{app:supplementary_methods}

\subsection{Algorithm}
Algorithm~\ref{app:alg:basic_ouija_dynamics} details the simulation procedure for the CoCre-Sam dynamics. This procedure involves a hierarchical control structure (illustrated in Fig.~\ref{app:fig:hierarchical_control}), including the high-level character selection loop and the low-level Langevin updates within each step.

\begin{algorithm}[h]
\caption{Basic Ouija Board Dynamics Simulation}
\label{app:alg:basic_ouija_dynamics}
\begin{algorithmic}[1]
\State Initialize context $c_{1:0} = ()$ (empty sequence), context length $\tau = 0$.
\State Set initial planchette position $x \leftarrow g_{\text{BOS}}$.  \Comment{Once per trial; the planchette is \emph{not} reset between characters.}
\Loop
    \State Set internal dynamics time step $t = 0$.
    \Repeat %
        \If{$t \% \triangle T = 0$}
        \State \textit{// (1) Next character prediction}
            \State Each agent $i \in \mathcal{I}$ determines its high-level action $\bar{c}_i \sim P_i(c_{\tau +1}|c_{1:\tau})$
        \EndIf
        \State \textit{// (2) Langevin dynamics on Ouija board}
        \State Each agent $i \in \mathcal{I}$ determines its low-level action $a_i(t)$ based on $x(t)$ and $E_{i}(x; \bar{c}_i)$.
        \State Calculate the collective action $a(t) = \sum_{i \in \mathcal{I}} a_i(t)$.
        \State Update planchette position: $x(t+1) = x(t) + a(t)$
        \State $t \leftarrow t + 1$.
        \If{$t > T_{burn\_in}$}
            \State Find the nearest goal $c^* = \arg\min_{c \in \mathcal{C}} \|x(t) - g_c\|$ and vote.
        \EndIf
    \Until{ $t \ge T_{max\_inner}$}
    \State Record the letter $c_{\tau+1} = c^*$ selected by voting.
    \State Update context $c_{1:\tau+1} = (c_{1:\tau}, c_{\tau+1})$ considering vote results.
    \State Update context length $\tau \leftarrow \tau + 1$.
    \If{$\tau \ge T_{max\_outer}$ or $c_{\tau} == \text{'EOS'}$}
        \State \textbf{break} Loop
    \EndIf
\EndLoop
\State \textbf{return} Generated sequence $c_{1:\tau}$.
\end{algorithmic}
\end{algorithm}

\subsection{Detailed Conditions of Experiment}
\label{app:subsec:experimental_conditions}

\subsubsection{Environment}
We simulate a 2D environment ($d=2$) representing the Ouija board, spanning a $[-1,7] \times [-1,4]$ region. The set of discrete outputs $\mathcal{C}$ includes 26 alphabet letters along with special tokens 'BOS' (Beginning of Sentence) and 'EOS' (End of Sentence), totaling $28$ symbols. Each symbol $s \in \mathcal{C}$ is assigned a goal location $g_c$ arranged on a uniform integer grid, with positions ${0, 1, \dots, 6} \times {0, 1, 2, 3}$ in unit spacing. If the planchette's continuous position $x(t)$ moves outside the defined region, it is clipped back to the boundary, i.e., constrained to remain within the $[-1,7] \times [-1,4]$ space. The initial planchette position evolves from the 'BOS' position, located at $g_{\text{BOS}} = [3, 0]$. See Fig.~\ref{fig:ouija_combined} (left) in the main text for the layout.

\subsubsection{Agents and Language Models} We primarily consider $N=2$ agents. Each agent $i$ is equipped with a simple internal language model $P_i(c_{\tau+1} \mid c_{1:\tau})$ representing their implicit knowledge or intention.
 For clarity and controlled experiments, we define these models based on simple n-gram models, rather than using large pre-trained LLMs directly.
This allows us to create scenarios with agents having identical and different internal models, i.e., preference of flowers.
$P_i(c_{\tau+1} \mid c_{1:\tau})$ is derived from simple $6$-gram statistics from a small corpus assigned to agent $i$.

We investigate the generation of meaningful sequences, specifically targeting flower names.
To define the set of possible outputs, we first constructed a vocabulary of $104$ distinct flower names (a curated A--Z pool) using ChatGPT-4o; these served as the reference set for our simulation.
To introduce differentiated preference distributions for Agent A and Agent B over this vocabulary, we generated $100{,}000$ samples for each agent via weighted random sampling.
The weights were determined based on the perceived colorfulness of each flower, as classified by ChatGPT-4o.
Agent~A was assigned higher sampling weights for flowers judged to be more visually colorful, resulting in a sampling bias toward vivid, vibrant flowers.
Let $w_A$ denote the normalized sampling weight assigned to a specific flower type for Agent A, based on its perceived colorfulness.
In contrast, Agent~B's weights were defined complementarily, using a reverse-weighting scheme proportional to $1 - w_A$ (normalized), thereby favoring less colorful or more subdued flower types.
This process enabled us to train distinct character-level $n$-gram language models ($n=6$) for each agent using their respective sampled flower names.

\subsubsection{Energy Function}
The individual energy $E_i(x \mid c_{1:\tau})$ for agent $i$ at position $x$, given the context $c_{1:\tau}$, is defined as the sum over all symbols $c \in \mathcal{C}$ of the product of the agent's probability for that symbol, $P_i(c \mid c_{1:\tau})$, and a potential $\phi(r)$. This potential $\phi(r)$ depends on the Euclidean distance $r = \|x - g_c\|$ between the current position $x$ and the goal location $g_c$ of symbol $c$ (see Eq.~\eqref{eq:elemental_energy} in the main text).

The specific potential function $\phi_c(r)$ is defined using the Cauchy (Lorentzian) loss function, which smoothly interpolates between a quadratic potential near the origin and a logarithmic growth at larger distances. This form ensures continuity and smoothness of both the potential and its derivative (interpreted as force) across all $r$. This function is often used in robust estimation~\cite{huber2011robust}. The function is defined as follows, where $r = \|x - g_s\|$ denotes the Euclidean distance from the planchette to goal $g_s$ (consistent with the definition above):
$$
\phi_\text{Cauchy}(r) = \frac{1}{2} \ln\left(1 + \left(\frac{r}{r_0}\right)^2\right)
$$
This formulation behaves like a spring potential for $r \ll r_0$ and suppresses long-range interactions for $r \gg r_0$.

Thus, the individual energy for agent $i$ and the fused energy $E_{\text{fused}}(x; c_{1:\tau})$ are as follows:
\[
E_i(x \mid c_{1:\tau}) = -\sum_{s \in \mathcal{C}} P_i(s \mid c_{1:\tau}) \phi_\text{Cauchy}(\|x-g_s\|), \quad E_{\text{fused}}(x \mid c_{1:\tau}) = \sum_{i=1}^N E_i(x \mid c_{1:\tau})
\]
In our experiments, the boundary radius is $r_0=0.3$.

\subsubsection{Dynamics Parameters}
Key parameters for the CoCre-Sam dynamics are the step size $\eta$ and the individual noise intensities $D_i$, which determine the effective temperature $\mathcal{T}_{\rm fused} = \sum D_i / \eta$.
In our experiments, we set the step size to $\eta = 0.1$ and the individual noise intensities to $D_i = 0.01$.
For $N = 2$ agents, this yields an effective temperature of $\mathcal{T}_{\rm fused} = 0.2$.
The inner simulation loop for selecting each character runs for $T_{\text{max\_inner}} = 2000$ steps.
To determine the character $c_{\tau+1}$ at the end of these steps, only the final $5\%$ of the trajectory (i.e., the last 100 steps) is used to vote for the symbol $s \in \mathcal{C}$ whose goal location $g_c$ is closest to the planchette's position.
The most frequently selected symbol is adopted as $c_{\tau+1}$.
The outer loop continues this process to generate a sequence, terminating either upon selecting the EOS symbol or upon reaching a maximum length of $20$ characters.

Two aspects of this minimal Ouija-board instance are worth noting. First, the Cauchy attractor ($r_0{=}0.3$) makes letter basins narrower than the unit inter-letter spacing, allowing the planchette to be transiently trapped between letter goals during a finite inner loop while the discrete tail vote quantises the continuous distribution. Second, the 2D lattice (Figure~\ref{fig:ouija_combined}) is \emph{frequency-agnostic}: goal locations sit on a uniform $7\times 4$ integer grid with unit spacing between adjacent cells, independent of how often the letters co-occur under $P_i$, thereby decoupling \emph{physical} from \emph{linguistic} proximity. Together these make the empirical character-level marginal a geometrically biased read-out of the asymptotic PoE guaranteed by Theorem~\ref{thm:sampling_correctness}---yet CoCre-Sam still attains uniformly low perplexity across all three evaluators (Table~\ref{tab:perplexity}). This robustness suggests a productive future direction: a \emph{learned} embedding that aligns physical and linguistic proximity, or an anisotropic potential whose basin radius depends on next-letter probabilities under $P_i$.

\subsection{Supplementary Framework Details}

\subsubsection{High-level Character Sequence Generation}

From a higher-level perspective, the dynamic process (Eq.~\eqref{eq:basic_dynamics} in the main text) aims to sequentially select symbols from $\mathcal{C}$ to form a message.
The sequence of symbols selected up to a certain point forms the context, $c_{1:\tau} = (c_1, c_2, \dots, c_\tau)$.
The selection of the next symbol, $c_{\tau+1}$, occurs when the planchette dynamics (Eq.~\eqref{eq:basic_dynamics} in the main text) cause $x(t)$ to converge near a specific goal location $g_{c^*}$.
Once $c^*$ is determined (e.g., by voting, see Algorithm~\ref{app:alg:basic_ouija_dynamics}), $c_{\tau+1} = c^*$.
This selected symbol is then appended to the context, $c_{1:\tau+1}$, influencing the subsequent actions of the agents via their internal predictive models $P_i(c_{\tau+1}|c_{1:\tau})$.
This process repeats, generating the message letter by letter, until a termination condition is met.

\subsubsection{Hierarchical Control and $\triangle T$ Discussion}

\begin{figure*}[tb]
  \centering
  \includegraphics[width=0.9\textwidth]{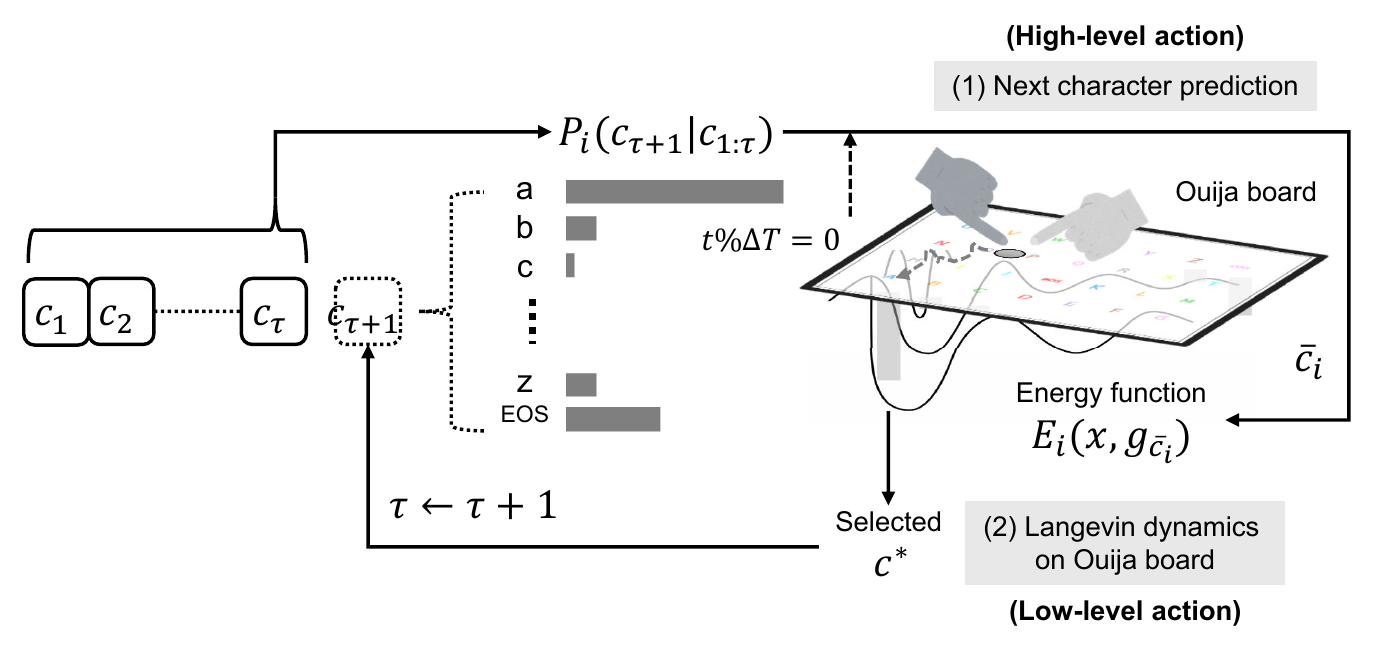}
  \caption{Flowchart of the hierarchical control in CoCre-Sam.
The outer loop for character sequence generation (see Algorithm~\ref{app:alg:basic_ouija_dynamics}) determines the energy function by weighting elemental energy functions based on the context $c_{1:\tau}$, while the inner loop simulates the probabilistic convergence of the planchette dynamics driven by the resulting energy landscape.}
  \label{app:fig:hierarchical_control}
\end{figure*}

We realize the probabilistic model $P_i$ as a temporally hierarchical stochastic decision-making process.
The process consists of two levels: a high-level action for goal selection and a low-level action for physical movement.
This hierarchical process is illustrated in Fig.~\ref{app:fig:hierarchical_control}.
The high-level action, performed by each agent $i$ every $\triangle T$ steps, involves sampling a temporary goal character $\bar{c}_i \sim P_i(c_{\tau+1} | c_{1:\tau})$ from its individual language model.
The low-level action is then determined by the elemental energy function $E_i(x; g_c)$ (Eq.~\eqref{eq:elemental_energy} in the main text) that guides the agent toward its target goal, $g_{\bar{c}_i}$.

The dynamics of this hierarchical process (detailed in Algorithm~\ref{app:alg:basic_ouija_dynamics}) depend on the time constant $1 \le \triangle T \le T_{max\_inner}$.
Two limiting cases illustrate its effect. When $\triangle T = T_{max\_inner}$, each agent commits to a single goal sampled at the beginning of the inner loop.
The dynamics then become a negotiation between these fixed individual goals.
In contrast, when $\triangle T = 1$, agents re-sample their goal at every time step.
The resulting rapid switching means the agent's movement is governed by an effective energy function (Eq.~\eqref{eq:effective_energy} in the main text) that averages over all possible goals, weighted by their probabilities.
Notably, character sequences generated through this hierarchical process with $\triangle T < T_{max\_inner}$ may deviate from those produced by the original language model $P_i(c_{\tau+1} \mid c_{1:\tau})$.

\section{Supplementary Theoretical Analysis}
\label{app:supplementary_theoretical_analysis}

\subsection{Assumptions}
\label{app:subsec:assumptions}

We make the following standard assumptions to support the theoretical analysis of Langevin-based sampling in the CoCre-Sam framework:

\begin{assumption}[Smoothness of Individual Energies]
\label{app:assump:smoothness}
For any given context $c_{1:\tau}$, each individual energy function $E_i(x \mid c_{1:\tau})$ (Eq.~\eqref{eq:effective_energy} in the main text) is continuously differentiable with respect to $x$. This ensures that the force terms derived from these energy functions are well-defined at all positions.
\end{assumption}

\begin{assumption}[Lipschitz Continuous Gradients]
\label{app:assump:lipschitz}
The gradient $\nabla_x E_i(x \mid c_{1:\tau})$ is Lipschitz continuous with respect to $x$, with Lipschitz constant $L_i$. This implies that the gradient of the fused energy function $E_{\rm fused}(x \mid c_{1:\tau}) = \sum_i E_i(x \mid c_{1:\tau})$ is also Lipschitz continuous, with constant $L = \sum_i L_i$. This property ensures that the gradient does not change too abruptly, which is important for the stability and convergence of Langevin dynamics.
\end{assumption}

\begin{assumption}[Noise Properties]
\label{app:assump:noise}
The noise terms $\xi_i(t)$ in the individual agent actions (Eq.~\eqref{eq:agent_action} in the main text) are independent and identically distributed (i.i.d.) standard $d$-dimensional Gaussian vectors, i.e., $\xi_i(t) \sim \mathcal{N}(0, I_d)$, and are independent across agents $i$ and time steps $t$. The noise intensities $D_i \ge 0$ are constant for each agent, and the total (or effective) noise magnitude is defined as $D_{\rm fused} = \sum_i D_i > 0$. This noise introduces stochasticity necessary for proper sampling behavior.
\end{assumption}

\begin{assumption}[Step Size]
\label{app:assump:step_size}
The step size $\eta$ (introduced in Eq.~\eqref{eq:agent_action} in the main text) used for the agent actions, which drive the basic dynamics (Eq.~\eqref{eq:basic_dynamics} in the main text), is a positive constant chosen sufficiently small to ensure numerical stability and convergence of the discrete-time Langevin updates.
\end{assumption}

These assumptions ensure that the gradients are well-defined and smoothly varying (Assump.~\ref{app:assump:smoothness}, \ref{app:assump:lipschitz}), the stochasticity is properly structured to support exploration (Assump.~\ref{app:assump:noise}), and the numerical integration scheme remains stable (Assump.~\ref{app:assump:step_size}). Collectively, they are mild and standard conditions for theoretical analysis of Langevin MCMC and are typically satisfied in energy-based models.

\subsection{Proof of Theorem 1: (CoCre-Sam Dynamics as Unadjusted Langevin Algorithm on Fused Energy)}

\begin{proof}[Proof Sketch]
We begin with the basic collective dynamics:
\[
x(t+1) = x(t) + \sum_{i \in \mathcal{I}} a_i(t).
\]
Substituting the definition of the agent action $a_i(t)$ from Eq.~\eqref{eq:agent_action} in the main text,
\[
a_i(t) = - \eta \nabla_x E_i(x(t) \mid  c_{1:\tau}) + \sqrt{2 D_i} \, \xi_i(t),
\]
we obtain:
\[
x(t+1) = x(t) + \sum_{i \in \mathcal{I}} \left( - \eta \nabla_x E_i(x(t) \mid c_{1:\tau}) + \sqrt{2 D_i} \, \xi_i(t) \right).
\]
Rewriting the terms, this becomes:
\[
x(t+1) = x(t) - \eta \sum_{i \in \mathcal{I}} \nabla_x E_i(x(t) \mid c_{1:\tau}) + \sum_{i \in \mathcal{I}} \sqrt{2 D_i} \, \xi_i(t).
\]
Using the definition $E_{\rm fused} = \sum_i E_i$, we have $\nabla_x E_{\rm fused} = \sum_i \nabla_x E_i$. Therefore:
\[
x(t+1) = x(t) - \eta \nabla_x E_{\rm fused}(x(t) \mid c_{1:\tau}) + \sum_{i \in \mathcal{I}} \sqrt{2 D_i} \, \xi_i(t).
\]

The third term is a sum of independent Gaussian noise vectors scaled by $\sqrt{2D_i}$. Since independent Gaussians remain Gaussian when summed, the entire noise term is distributed as a Gaussian with zero mean and covariance $2 D_{\rm fused} \, I_d$, where $D_{\rm fused} = \sum_i D_i$.

Letting $\xi(t) \sim \mathcal{N}(0, I_d)$, we can re-express the update as:
\begin{equation}
    x(t+1) = x(t) - \eta \nabla_x E_{\rm fused}(x(t) \mid c_{1:\tau}) + \sqrt{2 D_{\rm fused}} \, \xi(t).
    \label{app:eq:derived_cocresam_update}
\end{equation}

Now compare~\eqref{app:eq:derived_cocresam_update} with the standard Unadjusted Langevin Algorithm (ULA) update rule:
\begin{equation}
 x_{t+1} = x_t - \eta' \nabla E(x_t) + \sqrt{2 \eta' \mathcal{T}} \, \epsilon_t,
 \label{app:eq:ula}
\end{equation}
where $\eta'$ is the step size and $\mathcal{T}$ is the temperature.

Matching terms, we identify:
\[
\eta' = \eta, \quad E = E_{\rm fused}, \quad \text{and} \quad \sqrt{2 \eta' \mathcal{T}} = \sqrt{2 D_{\rm fused}}.
\]
Solving for $\mathcal{T}$ gives:
\[
\mathcal{T} = D_{\rm fused} / \eta.
\]

Thus, we conclude that the CoCre-Sam dynamics is mathematically equivalent to ULA applied to the fused energy landscape $E_{\rm fused}$, with an effective temperature $\mathcal{T}_{\rm fused} = D_{\rm fused}/\eta = \sum_i \mathcal{T}_i$.
\end{proof}

\subsection{Proof of Theorem 2: (Stationary Distribution and Sampling Correctness)}

\begin{proof}[Proof Sketch]
Theorem~\ref{thm:equivalence} shows that the CoCre\text{-}Sam dynamics is mathematically
equivalent to the Unadjusted Langevin Algorithm (ULA) applied to the fused energy
$E_{\mathrm{fused}}$, with an effective temperature
$\mathcal{T}_{\mathrm{fused}} = D_{\mathrm{fused}} / \eta $.
Assumptions~\ref{app:assump:smoothness}--\ref{app:assump:step_size} supply the standard
regularity conditions required for the convergence theory of ULA:
continuous differentiability, Lipschitz-continuous gradients, independent Gaussian
noise of strictly positive intensity, and a sufficiently small step size.
Under these conditions, ULA is \emph{geometrically ergodic} and converges to its
unique stationary Gibbs distribution
$P(x) \propto \exp\!\bigl(-E(x)/\mathcal{T}\bigr)$
\cite{RobertsTweedie1996, Dalalyan2017ULA}.
We note that the experimental implementation further confines the state to a compact domain via boundary clipping (Appendix~\ref{app:subsec:experimental_conditions}); on this compact set, the smoothness and Lipschitz-gradient assumptions are sufficient to guarantee geometric ergodicity, sidestepping the heavy-tailed convergence subtleties that the unbounded Cauchy potential would otherwise introduce on $\mathbb{R}^d$.

Applying the known results directly to our setting
(with $E = E_{\mathrm{fused}}$ and
$\mathcal{T} = \mathcal{T}_{\mathrm{\rm fused}}$) implies that the Markov chain
$\{x(t)\}_{t\ge 0}$ generated by CoCre\text{-}Sam converges
to the stationary distribution
\[
P^{\rm stat}_{\rm fused}(x)\;=\;
\frac{1}{Z}\exp\!\bigl(-E_{\mathrm{fused}}(x \mid c_{1:\tau})/\mathcal{T}_{\mathrm{\rm fused}}\bigr),
\]
where $Z$ is the normalising constant.
Hence, the collective dynamics indeed perform correct sampling from the probability
distribution implicitly defined by the fused energy landscape.
\end{proof}

\subsection{Derivation of Product of Experts Connection}
\label{app:subsec:poe_derivation}

Theorem~\ref{thm:sampling_correctness} exposes an explicit link between
CoCre\text{-}Sam and the \emph{Product of Experts} (PoE) framework:
\begin{align}
&\prod_{i \in \mathcal{I}} \{P^{\rm stat}_{\rm i}(x)\}^{{\mathcal{T}_i}/{\mathcal{T}_{\rm fused}}} \\
&\propto \prod_{i \in \mathcal{I}}\{\exp\left(-{E_{i}(x\mid c_{1:\tau})}/{ \mathcal{T}_{i}}\right)\}^{{\mathcal{T}_i}/{\mathcal{T}_{\rm fused}}}\\
&=\exp\left(\sum_{i \in \mathcal{I}}-{E_{i}(x\mid c_{1:\tau})}/{ \mathcal{T}_{\rm fused}}\right)\\
&=\exp\left( -{E_{\rm fused}(x\mid c_{1:\tau})}/{ \mathcal{T}_{\rm fused}}\right) \\
& \propto P^{\rm stat}_{\rm fused}(x)
\end{align}
where $\{P^{\rm stat}_{\rm i}\}$ is the $i$-th agent's unique stationary distribution.
In the special case $\mathcal{T}_i =\mathcal{T}_j = \mathcal{T}  \ (\forall i,j)$, $P^{\rm stat}_{\rm fused}(x) \propto \prod_{i \in \mathcal{I}} \{P^{\rm stat}_{\rm i}(x)\}^{ 1/N}$.

\section{Supplementary Background}
\label{app:supplementary_background}

\subsection{Haptic Interaction and Collective Motor Learning}

Going beyond mechanisms of synchronization of mutual motions
, physical interaction between participants, particularly through the haptic channel, plays a significant role in cooperative motor learning and collective decision-making. The principle that ``two is better than one'' has been demonstrated in motor learning studies; practicing with a partner haptically connected through a virtual spring has been shown to be more advantageous for learning a task in an unknown environment than practicing individually~\cite{Ganesh2014, Mireles2017}.

This collaborative benefit is further highlighted in tasks involving direct shared control. For instance, in an experiment where pairs of participants controlled a single virtual object under an unknown force field, novice-novice pairings fostered motor learning and adaptability more effectively than expert-novice collaborations~\cite{Saracbasi2021Haptic}. The authors suggest this is because peers develop coordination through a shared trial-and-error experience, rather than one party simply relying on an expert's guidance, exploring the free energy landscape of human-environment interactions more widely and effectively.

Furthermore, the temporal patterns within these haptic signals between paired participants are not merely incidental; paired participants under cooperative tasks can demonstrate collective decision-making and the emergence of joint action strategies. Crucially, such haptic communication may even be related to the development of higher-order cognitive functions like language~\cite{Thorne2019HapticStructure}.

While these perspectives address the physical origin of the movement and some of the psychological and social dynamics, a key open question remains: what is the underlying \textit{computational} process by which the decentralized, noisy micro-actions, reflecting participants' implicit knowledge (e.g., linguistic predictions), aggregate through physical interaction to produce structured, often language-like, emergent output?

\subsection{Language Model Fusion}

Combining multiple language models (LLMs) is an active research area aiming to create more robust, knowledgeable, or versatile models. Various explicit {\it language model fusion} techniques exist to achieve this.

One major approach involves operating in the parameter space by directly manipulating model weights. Examples include model averaging techniques like model soups \cite{Wortsman2022ModelSoups}, Fisher-weighted averaging using Fisher Mask Nodes \cite{Sharma2024FisherMask}, and task arithmetic for model editing or combination \cite{Ilharco2023TaskArithmetic}. These methods aim to combine the strengths of parent models, often fine-tuned on different data or tasks, into a single, unified parameter set.

Other methods operate in the output space or use dynamic model selection. Simpler forms include prediction ensembling, while more advanced techniques include SpecFuse for ensembling via next-segment prediction \cite{Meng2023SpecFuse}, the Metropolis-Hastings Captioning Game for decentralized knowledge fusion in vision-language models (VLMs) \cite{Matsui2025MHCG}, and FuseLLM for knowledge externalization and transfer between LLMs \cite{Wan2024FuseLLM}. Mixture-of-Experts (MoE) architectures are another form of dynamic selection, which route inputs to specialized experts (see \cite{Dou2024MoESurvey} for a survey; \cite{Fedus2022Switch} details Switch Transformers, a key sparsely activated MoE).

These studies all propose explicit, technical methods for fusing LLMs. In contrast, this paper proposes a different perspective. While CoCre-Sam's dynamics are computationally equivalent to an ensemble-based fusion, our primary contribution is to demonstrate how a real-world social phenomenon---the Ouija board---can be understood as an implicit and dynamic fusion of agents' internal language models, driven by embodied physical interaction.

\subsection{Langevin Dynamics for Sampling}

Langevin dynamics, originating from physics to describe particle motion under potential forces and random fluctuations \cite{Langevin1908,Lemons1997Langevin}, provides a powerful mathematical framework for sampling from complex probability distributions, particularly those defined via an energy function $E(x)$. In machine learning and statistics, its discretized forms are effective {Markov Chain Monte Carlo (MCMC)} methods for sampling from a target Gibbs distribution $p(x) \propto \exp(-E(x)/\tau)$, where $\tau$ is a temperature parameter (for a general overview of MCMC methods, see, e.g., \cite{Neal2011MCMC}).

A common discretization is the ULA. Early theoretical work explored discrete approximations of Langevin dynamics and their convergence properties \cite{RobertsTweedie1996}. The ULA and its variants, such as Stochastic Gradient Langevin Dynamics (SGLD) which incorporates mini-batch gradients, were later popularized in machine learning for Bayesian inference and sampling \cite{WellingTeh2011SGLD}. The ULA updates the state $x_t$ according to:
\begin{equation}
 x_{t+1} = x_t - \eta \nabla E(x_t) + \sqrt{2 \eta \mathcal{T}} \epsilon_t
 \label{app:eq:ula_general}
\end{equation}
Here, $\eta$ is a step size, $\nabla E(x_t)$ is the gradient of the energy function pushing the state towards lower energy regions, $\mathcal{T}$ is temperature, and $\epsilon_t \sim \mathcal{N}(0, I)$ is standard Gaussian noise introducing stochasticity for exploration. Under appropriate conditions on $E(x)$ and the step size $\eta$, iterating this update causes the distribution of $x_t$ to converge to the target distribution $P(x)$. Rigorous theoretical guarantees for the convergence of ULA and related Langevin MCMC algorithms have been established in various settings, analyzing aspects such as non-asymptotic convergence rates and behavior under different assumptions about the target distribution \cite{Dalalyan2017ULA}.
This makes Langevin MCMC a principled method for sampling when only the energy function (or its gradient) is accessible.

The principle of using gradient information combined with noise for sampling is fundamental in contexts like energy-based models (EBMs) \cite{LeCun2006EBM}.
In our CoCre-Sam framework, we leverage this core mechanism not for a single agent, but to model the collective dynamics emerging from multiple interacting agents.

\section{Supplementary Figures}
\label{app:supplementary_figures}

Figure~\ref{app:fig:generation_sequence_peony} provides a step-by-step visualization of the process, illustrating how the planchette stochastically navigates the evolving energy landscape to select each character of the word \texttt{PEONY}. The associated fused energy $E_{\text{fused}}(x(t)| c_{1:\tau})$ decreased on average over time, despite stochastic fluctuations, until the system settles into a local minimum, signifying the selection of a letter. This behavior is consistent with our theoretical analysis of the system as a collective Langevin sampling process.

\begin{figure*}[htb!p]
  \centering
  \begin{tabular}{cc}
    \includegraphics[width=0.45\textwidth]{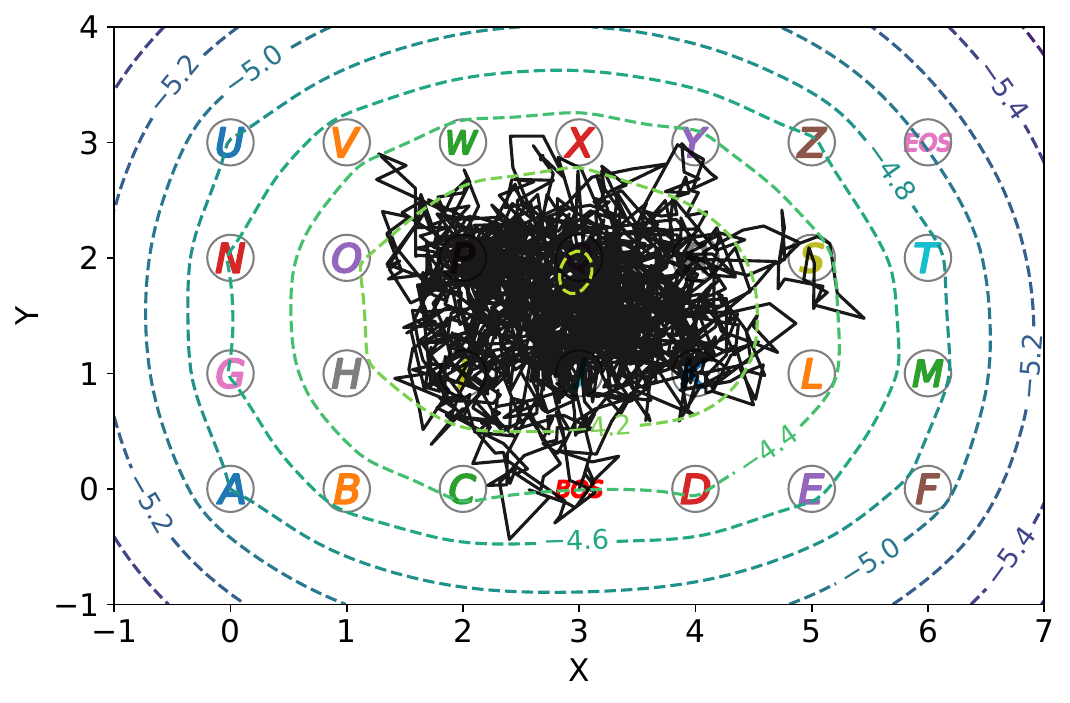} &
    \includegraphics[width=0.45\textwidth]{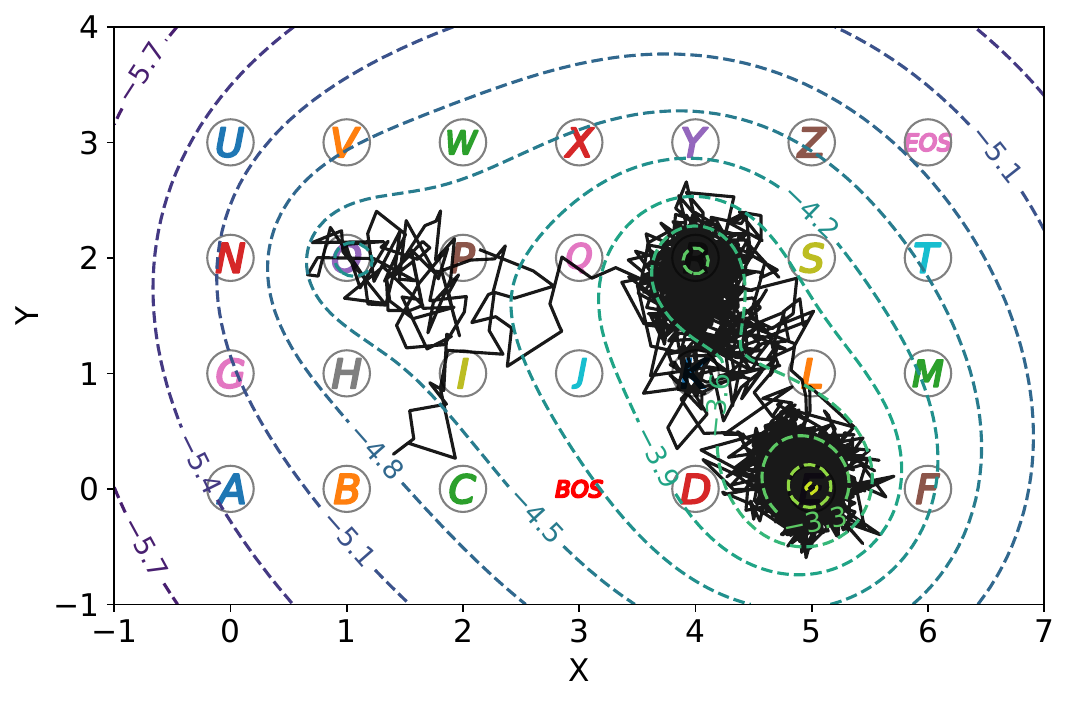} \\
    \includegraphics[width=0.45\textwidth]{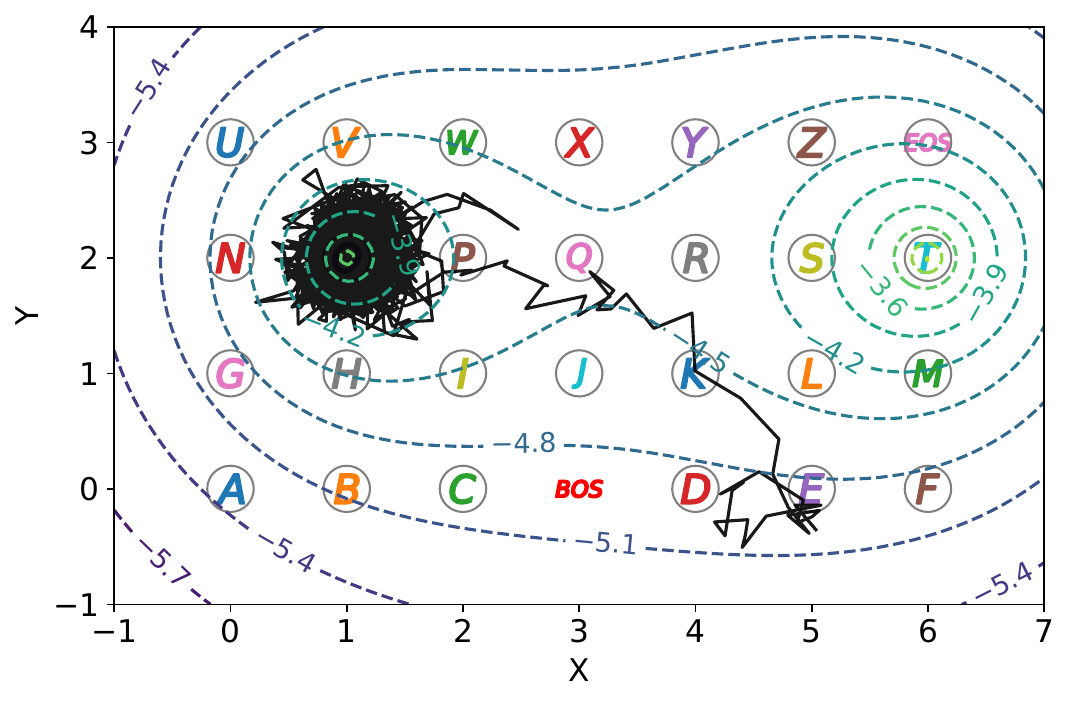} &
    \includegraphics[width=0.45\textwidth]{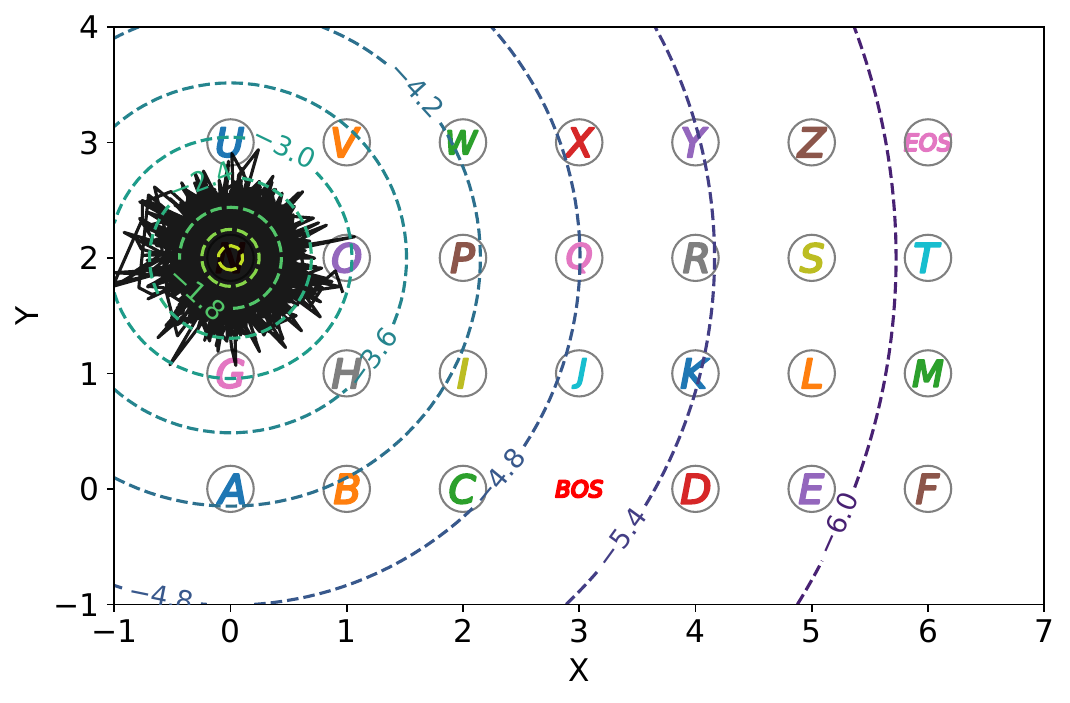} \\
    \includegraphics[width=0.45\textwidth]{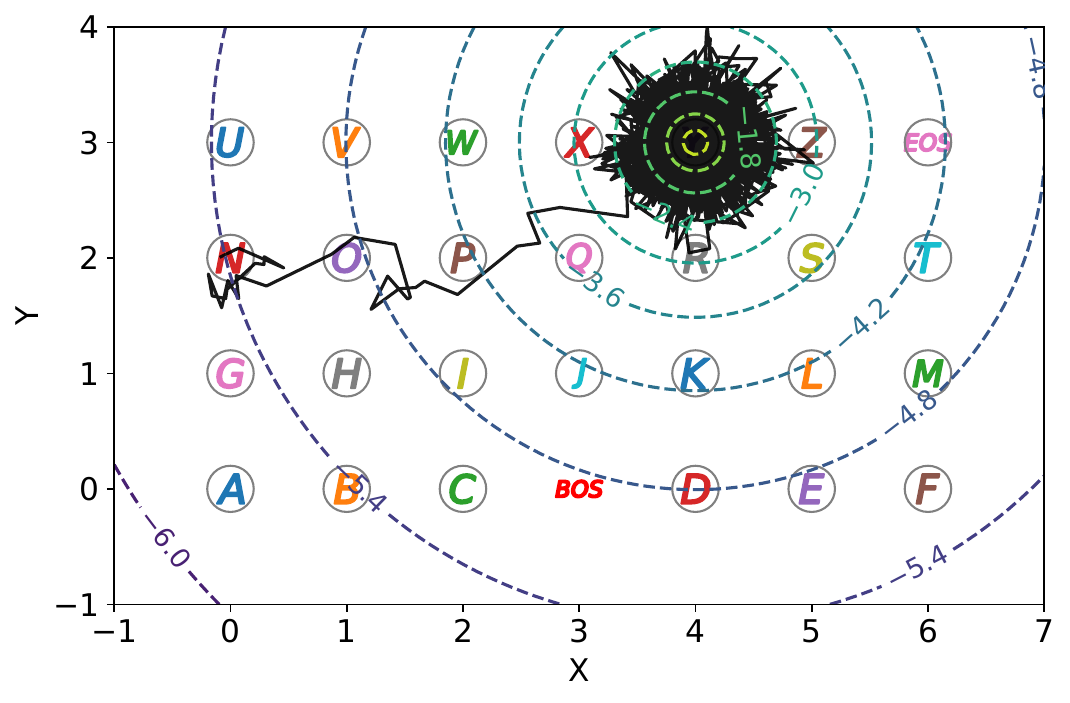} &
    \includegraphics[width=0.45\textwidth]{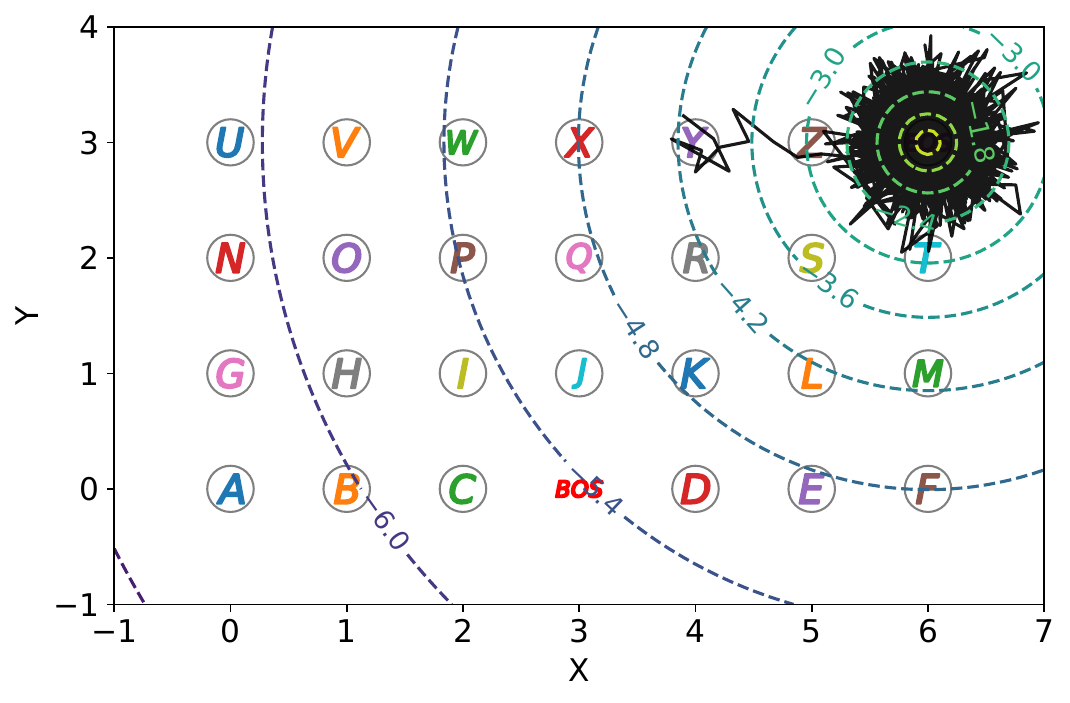} \\
  \end{tabular}
  \caption{Step-by-step visualization of the word generation process for the word \texttt{PEONY}, followed by the \texttt{[EOS]} token. The six panels, arranged left-to-right, then top-to-bottom in a 3$\times$2 layout, show the evolution of the fused energy landscape $E_{\rm fused}$ and the planchette's trajectory as each new character is selected.}
  \label{app:fig:generation_sequence_peony}
\end{figure*}

\end{document}